\documentclass[a4paper,UKenglish,cleveref, autoref, thm-restate]{lipics-v2021}

\title{Parallel Polynomial Permanent Mod Powers of 2 and Shortest Disjoint Cycles} 

\titlerunning{Parallel Polynomial Permanent} 

\author{Samir Datta}{Chennai Mathematical Institute, Chennai \and \href{sdatta@cmi.ac.in}{sdatta@cmi.ac.in} }{}{}{Partially funded by a grant from Infosys foundation and SERB-MATRICS grant MTR/2017/000480}

\author{Kishlaya Jaiswal}{Chennai Mathematical Institute, Chennai \and \href{kishlaya@cmi.ac.in}{kishlaya@cmi.ac.in}}{}{}{}

\authorrunning{S. Datta and K. Jaiswal} 

\Copyright{Samir Datta and Kishlaya Jaiswal} 



\category{} 

\relatedversion{} 

\acknowledgements{We would like to thank Eric Allender for a discussion regarding what was known about field operations in $\parity \L$. We would also like to thank Partha Mukhopadhyay for his comments on a preliminary version of the paper.}

\nolinenumbers 

\EventEditors{John Q. Open and Joan R. Access}
\EventNoEds{2}
\EventLongTitle{42nd Conference on Very Important Topics (CVIT 2016)}
\EventShortTitle{CVIT 2016}
\EventAcronym{CVIT}
\EventYear{2016}
\EventDate{December 24--27, 2016}
\EventLocation{Little Whinging, United Kingdom}
\EventLogo{}
\SeriesVolume{42}
\ArticleNo{23}

\usepackage{complexity, amsmath, mathtools, tikz}

\newcommand{\Modulo}[1]{\ (\mathrm{mod}\ #1)}
\newcommand{\perm}{\mathrm{perm}}
\newcommand{\sgn}{\mathrm{sgn}}
\newcommand{\hf}{\mathrm{hf}}
\newcommand{\pf}{\mathrm{pf}}
\DeclarePairedDelimiter\ceil{\lceil}{\rceil}
\DeclarePairedDelimiter\floor{\lfloor}{\rfloor}

\newfunc{\DLOGTIME}{DLOGTIME}

\begin{document}
	
	\maketitle
	
	\begin{abstract}
		We present a parallel algorithm for permanent mod $2^k$ of a matrix of univariate integer polynomials. It places the problem in $\parity \L \subseteq \NC^2$. This extends the techniques of Valiant \cite{DBLP:journals/tcs/Valiant79}, Braverman, Kulkarni and Roy \cite{DBLP:journals/cc/BravermanKR09} and Bj\"{o}rklund and Husfeldt \cite{DBLP:journals/siamcomp/BjorklundH19} and yields a (randomized) parallel algorithm for shortest 2-disjoint paths improving upon the recent (randomized) polynomial time algorithm \cite{DBLP:journals/siamcomp/BjorklundH19}.
		
		We also recognize the disjoint paths problem as a special case of finding disjoint cycles, and present (randomized) parallel algorithms for finding a shortest cycle and shortest 2-disjoint cycles passing through any given fixed number of vertices or edges.
	\end{abstract}

\section{Introduction}

The problem of computing the determinant of a matrix has been a very well studied problem in the past, and several fast (both sequential and parallel) algorithms are known. On the contrary, Valiant in his seminal paper \cite{DBLP:journals/tcs/Valiant79} showed that computing permanent of a matrix, an algebraic analogue of determinant, is hard. However modulo 2, determinant and permanent are equal and so building up on this, he presented an algorithm for computing permanent of an integer matrix modulo small powers of $2$. The algorithm uses Gaussian elimination which is known to be highly sequential and so it is desirable to have a parallel algorithm. This was resolved by \cite{DBLP:journals/cc/BravermanKR09} who presented a $\parity \L \subseteq \NC$ algorithm.

Moreover, $\NC$ algorithms for computing determinant of matrices over arbitrary commutative rings are also known, e.g. \cite{DBLP:journals/cjtcs/MahajanV97}. We would like to ask a similar question for the permanent. One natural extension would be to consider the ring of polynomials with integer coefficients. In this paper, we present a $\NC$ algorithm to compute permanent of matrices over integer polynomials modulo $2^k$ for any fixed $k$.

\begin{theorem}
	\label{thm:main1}
	Let $k \geq 1 $ be fixed and $A$ be a $n \times n$ matrix of integer polynomials, such that the degree of each entry is atmost $\poly(n)$. We can compute $\perm(A) \Modulo{2^k}$ in $\parity \L \subseteq \NC^2$
\end{theorem}

In the second half of the paper, we consider some applications of our parallel polynomial permanent algorithm. 
One direct consequence is that we are now able to parallelize the shortest 2-disjoint paths problem \cite{DBLP:journals/siamcomp/BjorklundH19}. Furthermore, we generalize this problem by adding more constraints on the disjoint paths - that the paths should pass through any given set of edges. This can also be viewed as a problem of finding 2-disjoint cycles, for which we present a randomized parallel algorithm, using the techniques from \cite{8efe47374b8c4cadb1165092ce46518d} and \cite{DBLP:journals/siamcomp/BjorklundH19}.

\begin{theorem}
	\label{thm:main2}
	Let $k \geq 1$ be fixed and $G$ be an undirected graph with $k$ marked vertices. We can find shortest 2-disjoint cycles passing through the marked vertices in $\parity \L / \poly$ (and $\RNC$)
\end{theorem}

Finally, we notice that a similar approach gives us an algorithm to compute Hafnians modulo $2^k$ of symmetric matrices of integers. Unfortunately, unlike the case of the permanent, we weren't able to extend this to a parallel algorithm. But nevertheless it gives a direct proof of the fact that counting number of perfect matchings modulo $2^k$, in any general graph, is in $\P$, as proved in \cite{DBLP:journals/cc/BravermanKR09}.

\subsection{Historical Survey}

The problem of computing permanent of an integer matrix was first shown to be $\NP$-hard by Valiant, where he also presented a $O(n^{4k-3})$ running time algorithm to compute permanent modulo $2^k$. It was also shown that computing permanent modulo any odd prime still remains hard. Zanko \cite{doi:10.1142/S0129054191000066} gave a proof for hardness of permanent under many-one reductions strengthening the result from the weaker Turing reductions used by Valiant. Later, Braverman, Kulkarni and Roy \cite{DBLP:journals/cc/BravermanKR09} presented a parallel $\parity \SPACE(k^2 \log n)$ algorithm for computing permanent modulo $2^k$. Bj\"{o}rklund and Husfeldt \cite{DBLP:journals/siamcomp/BjorklundH19} presented a $d^3n^{O(k)}$ time algorithm to compute permanent modulo $2^k$ of matrices over integer polynomials where the entries are of degree atmost $d$.

Finding $k$ disjoint paths in a graph has been a well studied problem in the past: given a graph (undirected/directed) and $k$ pairs of terminals $(s_i, t_i)_{1 \leq i \leq k}$, find $k$ pairwise vertex-disjoint paths $P_i$ from $s_i$ to $t_i$, if they exist.

When $k$ is not fixed (and is part of input) then the problem is known to be $\NP$-hard even for undirected planar graphs \cite{10.1145/1061425.1061430}. Linear time algorithms are known when further restricting directed planar graphs to the case: when all terminals lie on outer face \cite{DBLP:conf/soda/SuzukiAN90}, or when all the $s_i$-terminals lie on one common face while all the $t_i$-terminals lie on another common face \cite{DBLP:journals/ijfcs/Ripphausen-LipaWW96}. If we further ask for paths with minimal total length in the latter problem, then \cite{DBLP:journals/talg/VerdiereS11} presented a $O(kn \log n)$ running time algorithm to achieve the same.

When $k$ is fixed, the problem remains $\NP$-hard for directed graphs, even for $k=2$ \cite{DBLP:journals/tcs/FortuneHW80}, who had also given given a poly time algorithm for the restricted case of directed acyclic graphs. In the restriced setting of directed planar graphs, \cite{DBLP:journals/siamcomp/Schrijver94} presented a $n^{O(k)}$ running time algorithm, which was further improved to a fixed parameter tractable algorithm by \cite{DBLP:conf/focs/CyganMPP13}.

Shifting our focus to undirected graphs, the celebrated work of Robertson and Seymour \cite{ROBERTSON199565} gave a $O(n^3)$ algorithm for finding $k$ disjoint paths in an undirected graph, for any fixed $k$. \cite{DBLP:conf/fsttcs/DattaIK018} gave a parallel algorithm for class of planar graphs where all the terminals lie either on one or two faces. All this while, the question of finding \textit{shortest} disjoint paths in general undirected graphs, remained open for many years until recently, Bj\"{o}rklund and Husfeldt \cite{DBLP:journals/siamcomp/BjorklundH19} gave a polynomial time algorithm for finding the shortest 2-disjoint paths. For general $k$, this problem still remains open. Bj\"{o}rklund and Husfeldt also gave a parallel algorithm to count shortest 2-disjoint paths but only for cubic planar graphs \cite{DBLP:journals/corr/abs-1806-07586}.

\subsection{Our Techniques}

To compute permanent of a matrix $A$ over integer polynomials, we closely follow the analysis of \cite{DBLP:journals/cc/BravermanKR09} but immediately hit an obstacle. They give a reduction from $\perm(A) \Modulo{4}$ to several computations of $\perm(.) \Modulo{2}$, which crucially uses the fact that $\mathbb{Z}_2$ is a field. More precisely, when mimicking the proof, firstly it is required to find a non-trivial solution of $Av = 0$ with the property that atleast one of the entries of this vector is invertible. This fails\footnote{Let $A = \begin{pmatrix} x & x+1 \\ x & x+1 \end{pmatrix}$ then there does not exist any null vector of the form $\begin{pmatrix} f \\ 1 \end{pmatrix}$ or $\begin{pmatrix} 1 \\ f \end{pmatrix}$ for any $f \in \mathbb{Z}_2[x]$} over $\mathbb{Z}_2[x]$. Moreover, their algorithm also uses the fact that a non-singular matrix admits a LU decomposition iff all the leading principal minors are non-zero, which is known to hold in general only for matrices over fields.

Therefore, replacing $\mathbb{Z}$ with $\mathbb{Z}[x]$ in their analysis doesn't work as $\mathbb{Z}_2[x]$ isn't a field. Furthermore, any finite field $\mathbb{F}$ of characteristic 2 only corresponds to modulo 2 arithmetic. We need a way to extend the field structure so that it supports modulo $2^k$ arithmetic as well. If $\mathbb{F}$ was realized as $\mathbb{Z}_2[x]/(p(x))$ where $p(x)$ is irreducible over $\mathbb{Z}_2$ then a possible candidate is the ring $\mathbb{Z}[x]/(2^k,p(x))$. Therefore, the appropriate algebraic structure to consider would be the ring $\mathfrak{R} = \mathbb{Z}[x]/(p(x))$

Now we see that replacing $\mathbb{Z}$ with $\mathfrak{R}$ solves the above mentioned problems in the analysis, primarily because of the fact that $\mathfrak{R} \Modulo {2}$ is a finite char 2 field. With a slight bit of modification in the proof, we achieve that: given a matrix $A$ over $\mathfrak{R}$, we can find $\perm(A) \Modulo{2^k}$ or in other words if $A$ is a matrix over $\mathbb{Z}[x]$, we can compute $\perm(A) \Modulo{2^k, p(x)}$.

We are still not done because our aim was to compute $\perm(A) \Modulo{2^k}$ over $\mathbb{Z}[x]$. To achieve that, we choose $p(x)$ such that its degree is larger than the degree of polynomial $\perm(A)$. This requires doing computations over a large field. Alternatively, we develop a new way of interpolation over $\mathfrak{R}$, which allows us to choose $p(x)$ such that its degree is of logarithmic order of degree $\perm(A)$, but with a tradeoff of computing several (polynomially many) more permanents. We present this technique for its novelty.

Wahlström \cite{8efe47374b8c4cadb1165092ce46518d} addressed the question of finding a cycle passing through given vertices. We ask if we can also find shortest such cycle. And furthermore, can we also find shortest 2-disjoint cycles passing through these vertices? We combine techniques of \cite{8efe47374b8c4cadb1165092ce46518d} and \cite{DBLP:journals/siamcomp/BjorklundH19} to answer the above questions, by reducing them to computing permanents modulo 2 of $2^{k-1}$ and modulo 4 of $2^{k-1} + 2^{k-2}$ matrices respectively. These matrices are adjacency matrix of what we refer to as \textit{pattern} graphs. Notice that for $k=2$ finding shortest 2-disjoint cycles corresponds to finding shortest 2-disjoint paths (by connecting each pair of terminals with a common vertex), and in this case our pattern graphs are exactly those presented in \cite{DBLP:journals/siamcomp/BjorklundH19}.

\subsection{Organization of the Paper}

In section \ref{sec:prelim}, we first introduce the preliminaries and the notation that we shall be using throughout this paper. In the next section \ref{sec:genericPerm}, we present proof of our main theorem \ref{thm:main1} about computing permanent modulo $2^k$, following which we also discuss the complexity of required computations over the ring $\mathfrak{R}$ which shows that our algorithm is in $\parity \L$. We also present an alternative proof for our main theorem in section \ref{sec:interpolation} which uses new techniques. Then we present applications of our result that is finding shortest disjoint cycles in section \ref{sec:sdc}. Section \ref{sec:hafnian} discusses how to apply the same techniques to hafnian and hence it gives an alternate proof of the already known result that counting perfect matchings modulo $2^k$ is in $\P$.

\section{Preliminaries}
\label{sec:prelim}

We begin by stating the definition of the complexity class $\parity \L$.

\begin{definition}
	$\parity \L$ is the class of decision problems solvable by an $\NL$ machine such that
	\begin{itemize}
		\item If the answer is 'yes', then the number of accepting paths is odd.
		\item If the answer is 'no', then the number of accepting paths is even.
	\end{itemize}
\end{definition}

Given a $n \times n$ matrix $A = (a_{ij})_{i,j \in [n]}$, determinant and permanent of $A$ are defined as $$\det(A) = \sum_{\sigma \in S_n} \sgn(\sigma) \prod_{i=1}^n a_{i\sigma(i)} \quad\quad\quad \perm(A) = \sum_{\sigma \in S_n} \prod_{i=1}^n a_{i\sigma(i)}$$

The permanent of a matrix can be regarded as the weighted sum of cycle covers of an undirected graph. This gives a combinatorial interpretation to a seemingly pure algebraic quantity. We shall use this bridge to illustrate an application of our parallel polynomial permanent.

Let $G$ be a weighted undirected graph (not necessarily loopless) with the associated weight function $w$.

\begin{definition}
	We say $C \subseteq V(G) \times V(G)$ is a cycle cover of $G$ if
	\begin{itemize}
		\item $(u,v) \in C \implies \{u,v\} \in E(G)$
		\item $C$ is a union of vertex-disjoint simple directed cycles in $G$
		\item every vertex is incident to some directed edge in $C$
	\end{itemize}
\end{definition}

Note: loops are allowed as simple cycles in the above definition.

\begin{definition}
	For any cycle cover $C$ we denote the weight of $C$ by $\tilde{w}(C) = \prod_{e \in C} w(e)$
\end{definition}

The above definition is well-defined because any directed edge $(u,v)$ or $(v,u)$ in our cycle cover correspond to the same edge $\{u,v\}$ in our underlying undirected graph. And hence both these directed edges get the same weight, that is $w((u,v)) = w((v,u)) = w(\{u,v\})$. In literature, such type of weight functions are commonly referred to as \textit{symmetric} weight functions.

\begin{definition}
	Let $V(G) = [n]$ then we say $A = (a_{ij})_{i,j \in [n]}$ is the $(n \times n)$ adjacency matrix of $G$ if
	$$a_{ij} = \begin{cases}
	w(e) &\quad \text{if } e = \{i,j\} \in E(G) \\
	0 &\quad \text{otherwise}
	\end{cases}$$
\end{definition}

\begin{observation}
	$\perm(A) = \sum \tilde{w}(C)$ where the sum is taken over all cycle covers $C$ of $G$
\end{observation}

\section{Permanent over $\mathfrak{R}$ Mod $2^k$}
\label{sec:genericPerm}

To begin with, we fix some general notation. Let $p(x)$ be an irreducible polynomial over $\mathbb{Z}_2[x]$ such that $\deg(p(x))$ is atmost $\poly(n)$. Denote by $\mathbb{F}$ the finite field of char 2, which is realized as $\mathbb{Z}_2[x]/(p(x))$ and by $\mathfrak{R}_k = \mathbb{Z}[x]/(2^k, p(x))$. In particular $\mathfrak{R}_1 \cong \mathbb{F}$. Now as already discussed, we essentially replace $\mathbb{Z}$ by $\mathfrak{R}$ in the algorithm of \cite{DBLP:journals/cc/BravermanKR09}. We are ready to state the main theorem.

\begin{theorem}
	\label{lemma:genericPerm}
	Let $k \geq 1$ be fixed and $A \in \mathfrak{R}^{n \times n}$. We can compute $\perm(A) \Modulo{2^k}$ in $\parity \L$
\end{theorem}

Proof is by induction on $k$. We start with the base case $k=1$. Note that $\perm(A) \equiv \det(A) \Modulo{2}$. Using corollary \ref{lemma:determinantoverF} we can find $\perm(A) \Modulo{2}$ in $\parity \L$. Now suppose $k > 1$, we shall reduce it to computing several such determinants modulo 2, all of which can be computed in parallel. In doing so, we first illustrate an algorithm which is sequential and then we shall see how to parallelize it.

\subsection{Sequential algorithm for computing permanent modulo $2^k$}

We present the algorithm from \cite{DBLP:journals/cc/BravermanKR09} for computing permanent but translated within our framework. Let $A = (a_{ij})_{i,j \in [n]} \in \mathfrak{R}^{n \times n}$ be such that $\det(A) \equiv 0 \Modulo{2}$. Therefore we can find a non-zero vector $v \in \mathbb{F}^n$ such that $A^Tv = 0$ over $\mathbb{F}$. Assume without loss of generality $v_1 = 1$.

Let $r_i$ denote the $i^{th}$ row of $A$ and define $A'$ to be the matrix where the $1^{st}$ row in matrix $A$ is replaced with $\sum_i v_i r_i$. Now if we expand the permanent along the first row then we get
\begin{align}
\perm(A') = \sum_{i=1}^n v_i\perm(A[1 \leftarrow i]) = \perm(A) + \sum_{i=2}^n v_i\perm(A[1 \leftarrow i]) \label{eq:1}
\end{align}
where $A[i \leftarrow j]$ is the matrix $A$ but with $i^{th}$ row replaced with the $j^{th}$ row. For $I,J \subseteq [n]$ denote by $A[\widehat{I},\widehat{J}]$ the matrix obtained from $A$ by deleting rows indexed by $I$ and columns indexed by $J$. With this equation, modulo $2^k$ computation reduces to modulo $2^{k-1}$ computations of the minors as follows:
\begin{align*}
\perm(A') = \sum_{j=1}^n \left(\sum_{i=1}^n v_i a_{ij} \right) \perm(A[\widehat{\{1\}},\widehat{\{j\}}])
\end{align*}
Since $A^Tv=0 \Modulo{2}$, we can write $\sum_{i} v_i a_{ij} = 2b_j \Modulo{2^k}$ for some $b_j \in \mathfrak{R}_k$,  therefore, we can re-write the above permanents as:
\begin{align*}
\perm(A') \Modulo{2^k} &= 2 \left(\sum_{j=1}^n b_j \perm(A[\widehat{\{1\}},\widehat{\{j\}}]) \Modulo{2^{k-1}} \right)
\end{align*}
Similarly, expanding $\perm(A[1 \leftarrow i])$ along the $1^{st}$ and $i^{th}$ rows, we get the reduction:
\begin{align*}
&\perm(A[1 \leftarrow i]) = \sum_{j \neq k} a_{ij}a_{ik} \perm(A[\widehat{\{1,i\}},\widehat{\{j,k\}}]) \\
&\perm(A[1 \leftarrow i]) \Modulo{2^k} = 2 \left(\sum_{j < k} a_{ij}a_{ik} \perm(A[\widehat{\{1,i\}},\widehat{\{j,k\}}]) \Modulo{2^{k-1}}\right)
\end{align*}
Substituting these equations back in \ref{eq:1}, we get
\begin{equation*}
\begin{split}
\perm(A) \Modulo{2^k} &= 2 \left(\sum_{j=1}^n b_j \perm(A[\widehat{\{1\}},\widehat{\{j\}}]) \Modulo{2^{k-1}} \right) \\
& - 2 \sum_{i=2}^n v_i \left(\sum_{\substack{j,k = 1 \\ j<k}}^n a_{ij}a_{ik} \perm(A[\widehat{\{1,i\}},\widehat{\{j,k\}}]) \Modulo{2^{k-1}}\right)
\end{split}
\end{equation*}

Since addition and multiplication over $\mathfrak{R}_k$ is in $\parity \L$ (see \ref{lemma:arithmetic}) we get that $\perm(A) \Modulo{2^k}$ $\parity \L$-reduces to $\perm(.) \Modulo{2^{k-1}}$. Hence by induction, we can compute $\perm(A) \Modulo{2^k}$ in $\parity \L$, provided that $\perm(A) \equiv 0 \Modulo{2}$.

Let us see how to drop this assumption. We expand the permanent of $A$ along the $i^{th}$ row, then 
\begin{align*}
\perm(A) = \sum_{j} a_{ij} \perm(A[\widehat{\{i\}},\widehat{\{j\}}])
\end{align*}
If $\perm(A) \not \equiv 0 \pmod{2}$, then $\exists i,j$ such that $\perm(A[\widehat{\{i\}},\widehat{\{j\}}]) \not \equiv 0 \pmod{2}$.
Consider the matrix $C$ where all entries are same as $A$ except the $(i,j)^{th}$ entry which is replaced with $a_{ij}+y$. Then, we get $\perm(C) = \perm(A) + y\perm(A[\widehat{\{i\}},\widehat{\{j\}}])$. Notice that $\perm(A) + y\perm(A[\widehat{\{i\}},\widehat{\{j\}}]) \equiv 0 \pmod{2}$ is a linear equation in $y$ over the field $\mathbb{F}$ and so there exists a unique $y$ which satisfies this equation, which is $y_0  = \perm(A)\perm(A[\widehat{\{i\}},\widehat{\{j\}}])^{-1}$. Setting $y = y_0$ we get $\perm(C) \equiv 0 \Modulo{2}$, so we can compute $\perm(C) \pmod{2^k}$ and then compute $\perm(A[\widehat{\{i\}},\widehat{\{j\}}])$ recursively as $A[\widehat{\{i\}},\widehat{\{j\}}]$ is a smaller $(n-1) \times (n-1)$ size matrix. Hence we obtain $\perm(A) = \perm(C) - y_0\perm(A[\widehat{\{i\}},\widehat{\{j\}}]) \pmod{2^k}$. This yields a sequential algorithm for computing permanent modulo $2^k$ over $\mathfrak{R}$.

\subsection{Parallel algorithm for computing permanent modulo $2^k$}

The bottleneck was finding $i,j$ such that $A[\widehat{\{i\}},\widehat{\{j\}}]$ is non-singular over $\mathbb{F}$. We fix this by again appealing to the fact that we are working over a field, and modifying $A$ such that all leading principal minors are non-zero. This modification essentially derives from the following fact.

\begin{theorem} (\cite{2005math......6382O} Corollary 1)
	Let $A$ be an invertible matrix over a field $\mathbb{F}$, then all leading principal minors are non-zero iff $A$ admits an $LU$ decomposition
\end{theorem}

Every invertible matrix admits a $PLU$ factorization \cite{2005math......6382O} so let $A = PLU$. Denote by $Q = P^{-1}$, then $QA = LU$. Since $Q$ is also a permutation matrix, we get that $\perm(QA) = \perm(A)$ (because permanent is invariant under row swaps). Therefore, it suffices to give a $\parity \L$ algorithm to find $Q$ so that we can replace $A$ by $QA$ which is an invertible matrix such that all leading principal minors are non-zero. Thus computing $\perm(A) \Modulo{2^k}$ reduces to the problem of computing (in parallel) permanent modulo $2^k$ of $n-1$ matrices with $\perm \equiv 0 \Modulo{2}$. This gives a $\parity \L$ algorithm to compute permanent modulo $2^k$ over $\mathfrak{R}$.

To find $Q$, we closely follow \cite{218278}. For each $1 \leq i \leq n$, let $A_i$ be the matrix formed from $A$ by only taking the first $i$ columns. Let $A_i^j$ matrix obtained from $A_i$ by only taking the first $j$ rows. We construct a set $S_i \subseteq [n]$ inductively as follows:

\begin{itemize}
	\item Base case: $l \in S_i$ if $rank(A_i^l) = 1$ and $rank(A_i^k) = 0$ for all $k < l$
	\item Include $j \in S_i$ iff $rank(A_i^j) = 1 + rank(A_i^{j-1})$
\end{itemize}

Since $rank(A_i) = i$, we get $|S_i|=i$. Furthermore note that $S_i \subset S_{i+1}$. So let $S_1 = \{s_1\}$ and for each $i \geq 2$, denote by $s_i \in S_i \setminus S_{i-1}$. Consider the following permutation $Q = (n,s_n)\ldots(2,s_2)(1,s_1)$. Thus $Q$ is our desired permutation, such that $QA$ has all leading principal minors non-zero.

As a corollary, we immediately get our desired result.

\begin{corollary} (Theorem \ref{thm:main1} restated)
	Given a $n \times n$ matrix $A = (a_{ij})_{i,j \in [n]}$ over $\mathbb{Z}[x]$ with $\deg(a_{ij})$ atmost $\poly(n)$, we can compute $\perm(A) \Modulo{2^k}$ in $\parity \L$ for any fixed $k \geq 1$
\end{corollary}

\begin{proof}
	Let $N = n \max \{\deg(a_{ij})\} + 1$ and choose $l = \ceil{\log_3 (N/2)}$. Consider $p(x) = x^{2.3^l} + x^{3^l} + 1$ which is irreducible over $\mathbb{Z}_2[x]$ (see \cite{van2013introduction} Theorem 1.1.28)
	
	Since $\deg(p(x)) \geq N > \deg(\perm(A))$, using this $p(x)$ in above theorem, we get $\perm(A) \Modulo{2^k}$ for any fixed $k$.
\end{proof}

\subsection{Complexity Analysis}

We discuss the complexity results for arithmetic operations over the ring $\mathfrak{R}_k$ and matrix operations over the field $\mathbb{F}$, which were required in our above algorithm. To begin with, we state a well-known fact about integer polynomials matrix multiplication modulo $2$. This shall form our basis for showing computations over $\mathbb{F}$ in $\parity \L$.

\begin{lemma} (Folklore \cite{DAMM})
	\label{lemma:matrixExpo01}
	Let $A_1, A_2, \ldots A_n \in \mathbb{Z}_2[x]^{n \times n}$ then the product $A_1A_2 \ldots A_n$ can be computed in $\parity \L$
\end{lemma}

To obtain an analogous result over $\mathbb{F}$ we first perform multiplication over $\mathbb{Z}_2[x]$ and then divide all entries by $p(x)$, using the following polynomial division, as demonstrated by Hesse, Allender and Barrington in \cite{HAB}, to get that iterated matrix product over $\mathbb{F}$ is in $\parity \L$

\begin{lemma}(\cite{HAB} Corollary 6.5)
	\label{lemma:polydivision}
	Given $g(x), p(x) \in \mathbb{Z}[x]$ of degree atmost $\poly(n)$, we can compute $g(x) \Modulo{p(x)}$ in $\DLOGTIME-\text{uniform } \TC^0 \subseteq \L$
\end{lemma}

In particular, it follows that given an irreducible polynomial $p(x)$ (over $\mathbb{Z}_2[x]$), then for any $g(x) \in \mathbb{Z}[x]$ of degree atmost $\poly(n)$ we can find $g(x) \Modulo{2^k, p(x)}$ in $\parity \L$, for any fixed $k \geq 1$.

\begin{corollary}
	\label{lemma:matrixExpoF}
	Let $A_1, A_2, \ldots A_n \in \mathbb{F}^{n \times n}$ such that the degree of each entry is atmost $\poly(n)$ then the product $A_1A_2 \ldots A_n$ can be computed in $\parity \L$
\end{corollary}

Our algorithm also requires computing inverse of non-zero elements. To compute inverse over $\mathbb{F}^*$ we adopt the techniques from Fich and Tompa \cite{715897, 10.1145/44483.44496}. Since $\mathbb{F} = \mathbb{Z}_2[x]/(p(x))$ with $N = \deg(p(x))$ which is atmost $\poly(n)$, then given $a \in \mathbb{F}^*$, we observe that $a^{-1} = a^{q-2}$ where $q = 2^N = |\mathbb{F}|$.

We interpret this equation over $\mathbb{Z}_2[x]$, that is we need to compute $a(x)^{q-2} \Modulo{p(x)}$ over $\mathbb{Z}_2[x]$. First we show how to compute $a(x)^2 \Modulo{p(x)}$. Construct the $N \times N$ matrix $Q$ whose $i^{th}$ row $(Q_{i,0}, Q_{i,1}, \ldots, Q_{i,N-1})$ is defined as:

$$\sum_{j=0}^{N-1} Q_{i,j} x^j = x^{2i} \Modulo{p(x)}$$

for each $0 \leq i \leq N-1$. Matrix $Q$ can be computed in $\parity \L$ using the divison lemma \ref{lemma:polydivision}. Then the elements of the row vector $(a_0, a_1, \ldots, a_{N-1})Q$ are the coefficients of $a(x)^2 \Modulo{p(x)}$ as explained in section 3 of \cite{10.1145/44483.44496}. Furthermore, the coefficients of $a(x)^{2^k} \Modulo{p(x)}$ are given by $(a_0, a_1, \ldots, a_{N-1})Q^k$. From lemma \ref{lemma:matrixExpo01} we get that $a(x)^{2^k} \Modulo{p(x)}$ can be computed in $\parity \L$, for any $k$ bounded by $\poly(n)$. Therefore. writing $q-2 = (c_0, c_1, \ldots, c_{N-1})$ in binary,
$$a(x)^{q-2} \Modulo{p(x)} = \prod_{i=0}^{N-1} a(x)^{c_i2^i} \Modulo{p(x)}$$
can be computed in $\parity \L$, which gives us $a^{-1}$.

Mahajan and Vinay \cite{DBLP:journals/cjtcs/MahajanV97} describe a way to reduce the computation of a determinant over a commutative ring to  a semi-unbounded logarithmic depth circuit with addition and multiplication gates  over the ring. In fact, the following is an easy consequence of their result:

\begin{lemma}(Mahajan-Vinay~\cite{DBLP:journals/cjtcs/MahajanV97})
	\label{lemma:determinant}
	Let $A \in R^{n \times n}$ be a matrix over a commutative ring.
	Then there exist $M \in R^{(2n^2)\times (2n^2)}$ and two vectors
	$a,b\in R^{2n^2}$ such that $det(A) = a^T M^n b$. Moreover, each entry of
	the matrix $M_{ij}$ and the vectors $a,b$ is one of the entries $A_{i',j'}$ 
	or a constant from $\{0,1\}$ and the mapping 
	$\phi$ where for every $(i,j) \in [2n^2]\times [2n^2]$,
	$\phi(i,j) \in A_{[n]\times [n]} \cup\{0,1\}$
	is computable in Logspace. 
\end{lemma}
\begin{proof}(Sketch)
	In \cite{DBLP:journals/cjtcs/MahajanV97}, given a matrix $A$ they construct a graph $H_A$ whose vertex set
	is $\{s,t_+,t_-\} \cup Q$ where $Q = \{[p,h,u,i] : p \in \{0,1\}, h,u \in [n],i \in \{0,\ldots,n-1\}\}$.
	Moreover, the edges are one of the following forms $(s,q),(q,q'), (q,t_+)$
	and $(q,t_-)$ where $q,q' \in Q$
	and have weights $w(q,q')$ that each depend on a single entry of $A$ or are
	one of the constants $0,1$. Moreover the mapping is very simple to describe.
	Let us focus on the induced subgraph $H_A[Q]$. Notice that $|Q| = 2n^3$
	and each ``layer'' of $H_A[Q]$ is identical. In other words,
	$e_i = \left<[p,h,u,i],[p',h',u',i+1]\right>$ is an edge in $H_A[Q]$ iff
	$e_j = \left<[p,h,u,j],[p',h',u',j+1]\right>$ is an edge in $H_A[Q]$ and 
	both have the same weights for every $i,j \in \{0,\ldots,n-1\}$. Thus define
	the matrix $M$ by putting $M_{[p,h,u],[p',h',u']}$ as the weight of any of the 
	edges $e_i$. 
	
	Finally to define $a,b$: let $a_{[n \bmod{2},1,1]} = 1$ and $a_q = 0$ for all 
	other $q$. $b_{[1,h,u]} = a_{uh}$ and $b_{[0,h,u]} = -a_{uh}$. The correctness
	of our Lemma then follows from the proof of Lemma~2 of \cite{DBLP:journals/cjtcs/MahajanV97}.
\end{proof}

Using above lemma \ref{lemma:determinant}, we reduce determinant over $\mathbb{F}$ to matrix powering over $\mathbb{F}$, which can be computed in $\parity \L$ using corollary \ref{lemma:matrixExpoF}. Hence we get

\begin{corollary}
	\label{lemma:determinantoverF}
	Let $A \in \mathbb{F}^{n \times n}$ then $\det(A)$ can be computed in $\parity \L$
\end{corollary}

Now we demonstrate two results: computing rank and a null vector a matrix over $\mathbb{F}$ in $\parity \L$. We use Mulmuley's algorithm \cite{DBLP:journals/combinatorica/Mulmuley87}, which requires finding determinant over the ring $\mathbb{F}[y,t]$, which reduces to matrix powering over $\mathbb{F}[y,t]$ by the above result. We shall further reduce this to matrix powering over $\mathbb{F}$ as follows: Let $R$ be an arbitrary commutative ring. We associate with each polynomial $f(x) = \sum_{i=0}^{d-1} f_ix^i \in R[x]$ a $d \times d$ lower-triangular matrix 

$$P(f) = \begin{bmatrix}f_0 \\ f_1 & f_0 \\ f_2 & f_1 & f_0 \\ \vdots & \vdots & \vdots & \ddots \\ f_{d-1} & f_{d-2} & f_{d-3} & \ldots & f_0 \end{bmatrix} \in R^{d \times d}$$

Suppose we have two polynomials $f(x)$ and $g(x)$ of degree $d_1$ and $d_2$ respectively. We can interpret them as degree $d_1 + d_2$ polynomials (with higher exponent coefficients as $0$). Then we have that $P(f+g) = P(f) + P(g)$ and $P(fg) = P(f) P(g)$.

\begin{theorem}
	\label{theorem:matrixmultiply}
	Let $R$ be any commutative ring and $A_1, A_2, \ldots A_n$ be $n \times n$ matrices over $R[x]$ such that the degree of each entry is atmost $\poly(n)$. Denote by $\mathbf{A} = \prod A_l$. There exists $\poly(n) \times \poly(n)$ matrices $B_1, B_2, \ldots B_n$ over $R$ such that the coefficient of $x^k$ in $\mathbf{A}_{ij}$ is equal to $\mathbf{B}_{\psi(i,j,k)}$ where $\mathbf{B} = \prod B_l$ and $\psi$ is logspace computable.
	
	In other words, iterated matrix multiplication over $R[x]$ is logspace reducible to iterated matrix multiplication over $R$.
\end{theorem}

\begin{proof}
	Let $N = n \max_{i,j,k \in [n]} \{\deg ((A_i)_{jk})\}$ where $(A_i)_{jk}$ denotes the $(j,k)^{\text{th}}$ entry of $A_i$. By our assumption, $N$ is atmost $\poly(n)$. Now for each $1\leq i \leq n$, compute the matrix $B_i \in R^{nN \times nN}$ obtained from $A_i$ by replacing each polynomial $(A_i)_{jk}$ with the $N \times N$ matrix $P((A_i)_{jk})$. These matrices $B_i$ can be computed in log space. Now the coefficient of $x^k$ in $\mathbf{A}_{ij}$ can be read from the entry $\mathbf{B}_{\psi(i,j,k)}$ where $\psi(i,j,k) = ((i-1)N + k + 1, (j-1)N + 1)$ is logspace computable. The correctness follows from our observation $P(fg) = P(f)P(g)$.
\end{proof}

\begin{remark*}
	This gives us an alternate proof of the fact that iterated matrix multiplication over $\mathbb{Z}_2[x]$ is in $\parity \L$, as it follows immediately from the definition of $\parity \L$ that iterated matrix multiplication over $\mathbb{Z}_2$ is in $\parity \L$.
\end{remark*}

\begin{lemma} \cite{DBLP:journals/combinatorica/Mulmuley87}
	\label{lemma:rankoverF}
	Let $A \in \mathbb{F}^{m \times n}$ then $rank(A)$ can be computed in $\parity \L$
\end{lemma}

\begin{proof}
	We can assume that $A$ is a square ($n \times n$) symmetric matrix because otherwise replace $A$ with $\begin{pmatrix}0 & A \\ A^T & 0\end{pmatrix}$ which has rank twice that of $A$. Let $Y$ be a $n \times n$ diagonal matrix with the $(i,i)^{th}$ entry as $y^{i-1}$. And let $m$ be the smallest number such that $t^m$ has a non-zero coefficient in the characteristic polynomial of $YA$, that is $\det(tI - YA)$. Then rank of $A = n-m$.
	
	Suffices to show that $\det(tI - YA)$ can be computed in $\parity \L$. Notice that $(tI - YA) \in \mathbb{F}[y,t]^{n \times n}$ and so $\det(tI-YA)$ is logspace reducible to matrix powering over $\mathbb{F}[y,t]$. Using the canonical isomorphism $\mathbb{F}[y,t] \cong \mathbb{F}[y][t]$, repeated application of theorem \ref{theorem:matrixmultiply} logspace reduces it to matrix powering over $\mathbb{F}$.
\end{proof}

\begin{observation}
	Let $A \in \mathbb{F}^{n \times n}$ be an invertible matrix then $A^{-1}$ can be computed in $\parity \L$
\end{observation}

This follows from the fact that computing $A^{-1}$ involves computing the determinant of $A$ and $n^2$ cofactors, that is determinants of $n^2$ matrices of size $(n-1) \times (n-1)$. Notice that this also requires inverting the determinant, an element of $\mathbb{F}^*$, which has been explained above.

\begin{corollary}
	\label{lemma:nullvector}
	Let $A \in \mathbb{F}^{n \times n}$ then finding a non-trivial null vector (if it exists) is in $\parity \L$
\end{corollary}

\begin{proof}
	Let $rank(A) = m$, then permute the rows and columns of $A$ so that we can express $A = \begin{pmatrix} B & C \\ D & E \end{pmatrix}$ such that $B$ is an invertible $m \times m$ matrix. Let $\begin{pmatrix}x \\ y\end{pmatrix}$ be a column vector where $x \in \mathbb{F}^m$ and $y \in \mathbb{F}^{n-m}$, such that
	$$A \begin{pmatrix}x \\ y \end{pmatrix} = 0 \implies \begin{pmatrix} B & C \\ D & E \end{pmatrix} \begin{pmatrix}x \\ y \end{pmatrix} = 0$$
	
	This reduces to the set of equations: $Bx+Cy=0$ and $Dx+Ey = 0$. But the later is a redundant set of equations because $\begin{pmatrix}D & E\end{pmatrix}$ can be written in terms of $\begin{pmatrix}B & C\end{pmatrix}$. More precisely, there exists a matrix $V \in \mathbb{F}^{n-m \times m}$ such that $D = VB$ and $E=VC$ and so $Dx+Ey = VBx + VCy = V(Bx+Cy) = 0$. Therefore setting $x = \boldsymbol{1}$ and $y = -B^{-1}C\boldsymbol{1}$, gives us the desired null vector. So it suffices to give a $\parity \L$ algorithm to transform $A$ to the form as specified above, which follows from \cite{218278}. Let $A_i$ be the matrix formed from first $i$ rows of $A$. We construct a set $S \subseteq [n]$ as follows:
	
	\begin{itemize}
		\item Base case: $i \in S$ if $rank(A_i) = 1$ and $rank(A_j) = 0$ for all $j < i$
		\item Include $k \in S$ iff $rank(A_k) = 1 + rank(A_{k-1})$
	\end{itemize}
	
	It follows that $|S| = m$ and let $S = \{i_1 < i_2 < \cdots < i_m\}$ and $P_r$ be the permutation matrix described by $(m,i_m)\ldots(2, i_2)(1, i_1)$. Then $P_rA$ is the required matrix having first $m$ rows as linearly independent. Next, consider the matrix $A' = (P_rA)^T$ and apply the above algorithm to get a permutation matrix $P_c$ such that first $m$ rows of $P_cA'$ are linearly independent. Then $P_rAP_c^T$ is the required matrix such that the leading principal $m-$minor is non-singular.
\end{proof}

Finally, to conclude our result, we discuss arithmetic over $\mathfrak{R}_k$

\begin{lemma}
	\label{lemma:arithmetic}
	Let $k \geq 1$ be fixed then the following operations can be done in $\parity \L$
	\begin{itemize}
		\item Multiplication : Given  $a,b \in \mathfrak{R}_k$ compute $ab$
		\item Iterated Addition: Given $c_1,c_2,\ldots,c_n \in \mathfrak{R}_k$ compute $\sum_i c_i$
	\end{itemize}
\end{lemma}

\begin{proof}
	We use the fact that the arithmetic operations mentioned in the statement of lemma, but over $\mathbb{Z}_{2^k}$ are in $\parity \L$ (see for e.g. \cite{HAB})
	
	\begin{itemize}
		\item Let $a(x), b(x) \in \mathfrak{R}_k$ and write $a(x) = \sum_{i=0}^{D} a_i x^i$ and $b(x) = \sum_{i=0}^{D'} b_i x^i$, then $a(x)b(x) = \sum_{i=0}^{D+D'} \left(\sum_{j=0}^i a_jb_{i-j}\right) x^i$. Finally, using lemma \ref{lemma:polydivision}, divide $a(x)b(x)$ by $p(x)$ to obtain $ab \in \mathfrak{R}_k$
		\item Similarly, let $c_1(x), c_2(x), \ldots, c_n(x) \in \mathfrak{R}_k$ and write $c_i(x) = \sum_{j=0}^{D_i} c_{ij} x^j$ for each $i \in [n]$, then $\sum_{i=1}^n c_i(x) = \sum_{j=0}^{\max\{D_i\}} \left(\sum_{i=1}^n c_{ij}\right) x^j$ where we assume $c_{ij} = 0$ if $j > D_i$
	\end{itemize}
\end{proof}

\subsection{Examples}

Let 
$A = \begin{pmatrix}
1 & x+1 & x+2 \\
x & x^2 & x^2+x \\
x^2 & 3 & x^2+3
\end{pmatrix}$, 
$p(x) = x^6 + x^3 + 1$ be the irreducible polynomial and we want to evaluate $\perm(A) \Modulo{4}$ over the ring $\mathfrak{R} = \mathbb{Z}[x]/(p(x))$. First of all, a direct computation gives us $\perm(A) = 2x^5+6x^4+2x^3+12x^2+12x$. Now we demonstrate the steps taken by our algorithm.

\textbf{Step 1}: We start by evaluating $\perm(A) \Modulo{2}$. We directly notice here that last column is the sum of first two columns and so $\det(A) = 0 \implies \perm(A) \equiv 0 \Modulo{2}$
\newline

\textbf{Step 2}: Since $\det(A) \equiv 0 \Modulo{2}$, we solve the equation $A^Tv=0$ over $\mathbb{F}$ by our method as follows:
$\begin{pmatrix}
1 & x & x^2 \\
x+1 & x^2 & 1 \\
x & x^2+x & x^2+1
\end{pmatrix}
\begin{pmatrix}
v_1 \\
v_2 \\
v_3
\end{pmatrix} = 0$ 

Since rank of the principal $2 \times 2$ submatrix is already $2$, we set $v_3 = 1$ and solve the equation: 
$\begin{pmatrix} v_1 \\ v_2 \end{pmatrix} = - \begin{pmatrix}
1 & x \\
x+1 & x^2
\end{pmatrix}^{-1}
\begin{pmatrix}
x^2 \\ 1
\end{pmatrix} 1$
to get $v_1 = x^3+1$ and $v_2 = x^5+x$.
\newline

\textbf{Step 3}: For each $j=1,2,3$, we find $b_j$ such that $\sum_i v_i a_{ij} = 2b_j \Modulo{4}$

\begin{align*}
j=1: &&(x^3+1)+x(x^5+x)+x^2 &= 2x^2 \\
j=2: &&(x+1)(x^3+1)+x^2(x^5+x)+3 &= 2x^3 \\
j=3: &&(x+2)(x^3+1)+(x^2+x)(x^5+x)+x^2+3 &= 2x^3 + 2x^2
\end{align*}

\textbf{Step 4}: We have the formula
\begin{equation*}
\begin{split}
\perm(A) \Modulo{4} &= 2 \left(\sum_{j=1}^3 b_j \perm(A[\widehat{\{3\}},\widehat{\{j\}}]) \Modulo{2} \right) \\
& - 2 \sum_{i=1}^2 v_i \left(\sum_{\substack{j,k = 1 \\ j<k}}^3 a_{ij}a_{ik} \perm(A[\widehat{\{3,i\}},\widehat{\{j,k\}}]) \Modulo{2}\right)
\end{split}
\end{equation*}

\textbf{Step 4.1}:
\begin{align*}
\perm(A[\widehat{\{3\}},\widehat{\{1\}}]) &= \perm\begin{pmatrix} x+1 & x+2 \\ x^2 & x^2+x \end{pmatrix} = x \Modulo{2} \\
\perm(A[\widehat{\{3\}},\widehat{\{2\}}]) &= \perm\begin{pmatrix} 1 & x+2 \\ x & x^2+x \end{pmatrix} = x \Modulo{2} \\
\perm(A[\widehat{\{3\}},\widehat{\{3\}}]) &= \perm\begin{pmatrix} 1 & x+1 \\ x & x^2 \end{pmatrix} = x \Modulo{2} \\
\implies \sum_{j=1}^3 b_j \perm(A[\widehat{\{3\}},\widehat{\{j\}}]) &= ((x^3) + (x^4) + (x^4+x^3)) = 0 \Modulo{2}
\end{align*}

\textbf{Step 4.2}:
\begin{equation*}
\begin{split}
\sum_{\substack{j,k = 1 \\ j<k}}^3 a_{1j} & a_{1k} \perm(A[\widehat{\{1,3\}},\widehat{\{j,k\}}]) \\
&= (x+1)(x^2+x) + (x+2)x^2 + (x+1)(x+2)x = x^3+x^2+x \Modulo{2} \\
\sum_{\substack{j,k = 1 \\ j<k}}^3 a_{2j} & a_{2k} \perm(A[\widehat{\{2,3\}},\widehat{\{j,k\}}]) \\
&= x^3(x+2) + x(x^2+x)(x+1) + x^2(x^2+x) = x^4 + x^3 + x^2 \Modulo{2}
\end{split}
\end{equation*}
\begin{align*}
\sum_{i=1}^2 v_i \left(\sum_{\substack{j,k = 1 \\ j<k}}^3 a_{ij}a_{ik} \perm(A[\widehat{\{3,i\}},\widehat{\{j,k\}}]) \Modulo{2}\right) = x^5+x^4+x^3 \Modulo{4}
\end{align*}

Therefore, $\perm(A) \Modulo{4} = 2x^5+2x^4+2x^3$ which matches with our direct computation.

\paragraph*{Example 2}

Consider $A = \begin{pmatrix} 1 & x & x^2 \\ x & x^2 & 1 \\ 1 & x^2 & x \end{pmatrix}$, and so $\perm(A) = x^5+x^4+x^2+x$. Therefore, we now have $\det(A) \not \equiv 0 \Modulo{2}$

\textbf{Step 1}: Find $Q$ such that $QA$ has all leading principal minors are non-zero. In this case, we will get $Q = \begin{pmatrix} 1 & 0 & 0 \\ 0 & 0 & 1 \\ 0 & 1 & 0 \\  \end{pmatrix} \implies QA = \begin{pmatrix} 1 & x & x^2 \\ 1 & x^2 & x \\ x & x^2 & 1 \end{pmatrix}$
\newline

\textbf{Step 2}: Consider matrix $C$ whose all entries are same as $A$ except the last one which is incremented by $y$, that is $C = \begin{pmatrix} 1 & x & x^2 \\ 1 & x^2 & x \\ x & x^2 & 1+y \end{pmatrix}$, then $\perm(C) = \perm(A) + y\perm(A[\widehat{\{3\}}, \widehat{\{3\}}])$. Again construct $C'$ same as $A[\widehat{\{3\}}, \widehat{\{3\}}]$ but replace the last entry incremented by $y'$, that is $C' = \begin{pmatrix} 1 & x \\ x & x^2 + y' \end{pmatrix} \implies \perm(C') = \perm(A[\widehat{\{3\}}, \widehat{\{3\}}]) + y'\perm(A[\widehat{\{2,3\}}, \widehat{\{2,3\}}])$. Written as one equation, we get
\begin{align*}
\perm(A) = \perm(C) - y\left(\perm(C') - y'a_{11}\right)
\end{align*}
In this equation, both $C, C'$ are matrices with $\det \equiv 0 \Modulo{2}$ with the correct choice of $y, y'$, which were: 
\begin{align*}
y_0 = \perm(A)\perm(A[\widehat{\{3\}}, \widehat{\{3\}}])^{-1} \Modulo{2} = (x^5+x^4+x^2+x)(x^4+x^3+x^2) = x^3+1 \\ 
y'_0 = \perm(A[\widehat{\{3\}}, \widehat{\{3\}}])\perm(A[\widehat{\{2,3\}}, \widehat{\{2,3\}}])^{-1} \Modulo{2} = x^2+x
\end{align*}

So we can compute $\perm(C)$ and $\perm(C')$ by above method and substitute it into previous equation to get $\perm(A) \Modulo{4}$

\section{Permanent via Interpolation}
\label{sec:interpolation}

We now demonstrate another technique to compute permanent modulo $2^k$, which doesn't resort to computations over exponentially sized fields. This proceeds by choosing small degree polynomial $p(x)$. The techniques developed in this section are new and hence interesting by themself.

First we mention a result from \cite{DBLP:journals/corr/abs-2011-03819} used to interpolate the coefficients of a polynomial.

\begin{lemma} (\cite{DBLP:journals/corr/abs-2011-03819} Lemma 3.1)
	Let $\mathbb{F}$ be a finite, characteristic $2$, field of order $q$.
	$$\sum_{a \in \mathbb{F}^*} a^m = \begin{cases}
	0 &\quad \text{if } q-1 \nmid m \\
	1 &\quad \text{otherwise}
	\end{cases}$$
\end{lemma}

This dichotomy allows us to extract coefficients of any integer polynomial.

\begin{lemma} (\cite{DBLP:journals/corr/abs-2011-03819} Corollary 3.2)
	Let $f(x) = \sum_{i=0}^d c_ix^i$ be a polynomial with integer coefficients and $q > d+1$, then for any $0 \leq t \leq d$,
	$$\sum_{a \in \mathbb{F}^*} a^{q-1-t}f(a) = c_t \Modulo{2}$$
\end{lemma}

But this gives us the coefficients modulo $2$ only. How do we get coefficients modulo $2^k$?

The crucial observation here is that the above sum was computed over $\mathbb{F}$. So instead we do so over $\mathfrak{R}$ by identifying a copy of $\mathbb{F}^* \xhookrightarrow{} \mathfrak{R}$, and then we have
$$\sum_{a \in \mathbb{F}^*} a^m = \begin{cases}
2\alpha_m &\quad \text{if } q-1 \nmid m \\
2\beta_m + 1 &\quad \text{otherwise}
\end{cases}$$
where $\alpha_m, \beta_m \in \mathfrak{R}$.

Now we use repeated squaring method to obtain our desired modulo $2^k$ result
\begin{lemma}
	$\forall m \geq 0, k \geq 1$
	$$\sum_{a_1,\ldots,a_{2^{k-1}} \in \mathbb{F}^*} (a_1a_2 \cdots a_{2^{k-1}})^m = \left(\sum_{a \in \mathbb{F}^*} a^m\right)^{2^{k-1}} = 
	\begin{cases}
	0 \Modulo{2^k} &\quad \text{if } q-1 \nmid m \\
	1 \Modulo{2^k} &\quad \text{otherwise}
	\end{cases}$$
\end{lemma}

Note: We remind the reader that the computation here is done over $\mathfrak{R}_k$

\begin{proof}
	Fix any $m \geq 0$. Clearly $(2\alpha_m)^{2^{k-1}} \equiv 0 \Modulo{2^k}$.
	
	Suffices to show $(2\beta_m+1)^{2^{k-1}} \equiv 1 \Modulo{2^k}$.  This follows from induction on $k$. For $k=1$ the result holds as stated above. Assume that for some $k \geq 1$, the result holds. Then we have $(2\beta_m+1)^{2^{k-1}} = 2^k\gamma_{m,k} + 1$ where $\gamma_{m,k} \in \mathfrak{R}$
	
	$$(2\beta_m+1)^{2^k} =  \left((2\beta_m+1)^{2^{k-1}}\right)^2 = (2^k\gamma_{m,k}+1)^2 = 1 \Modulo{2^{k+1}}$$
\end{proof}
Using this we can interpolate coefficients of an integer polynomial as follows:

$$\sum_{a_1, \ldots, a_{2^{k-1}} \in \mathbb{F}^*} (a_1 \ldots a_{2^{k-1}})^{q-1-t}f(a_1 \ldots a_{2^{k-1}}) = c_t \Modulo{2^k}$$
Finally let $A(x)$ be a $n \times n$ matrix of integer polynomials and the permanent polynomial be 

$$\perm(A(x)) = \sum_{i=0}^N c_i x^i$$
From the above lemma it follows that

$$\sum_{a_1, a_2 \ldots \in \mathbb{F}^*} (a_1a_2 \cdots)^{q-1-t}\perm(A(a_1a_2 \cdots)) = c_t \Modulo{2^k}$$
provided that our field $\mathbb{F}$ is of order atleast $N+2$. For this, fix $l = \ceil{\frac{\log \log N}{\log 3}}$ so that $2^{2.3^l} > N+1$. Hence the field obtained from the irreducible polynomial $p(x) = x^{2.3^l} + x^{3^l} + 1$ (\cite{van2013introduction} Theorem 1.1.28) serves the purpose. It suffices to compute $|\mathbb{F}^*|^{2^{k-1}} = O(N^{2^{k-1}})$ many permanents over $\mathfrak{R}_k$ to obtain all the coefficients $c_t$ modulo $2^k$, all of which can be computed in parallel. Hence, we can compute the permanent of $A$ modulo $2^k$ in $\parity \L$.

\section{Shortest Disjoint Cycles}
\label{sec:sdc}

Now that we have a $\parity \L$ algorithm to compute permanent mod $2^k$, we are all set to demonstrate a parallel algorithm for shortest 2-disjoint paths. But we notice that we can place this problem in a more general framework of shortest disjoint cycles. Let us first formally define these problems.

$SDP(k)$: Given a weighted undirected graph with $k$ pairs of marked vertices $\{(s_i, t_i) \mid 1 \leq i \leq k\}$, find the minimum of sum of weight of paths between each pair $s_i$ and $t_i$ such that all paths are pairwise disjoint.

$SDC(l,k)$: Given a weighted undirected graph with $k$ marked vertices, find the minimum of sum of weight of $l$ cycles such that they pass through all the marked vertices and are pairwise disjoint and each cycle is incident to atleast one of the marked vertices.

\begin{note*}
	We only consider \textit{non-trivial cycles} that is we don't consider self-loops in the above problem.
\end{note*}

Given an instance of the $SDP(2)$ problem, join the pairs of vertices $(s_1,t_1)$ and $(s_2,t_2)$ with new vertices $u_1$ and $u_2$ respectively, as show in Figure \ref{fig:join}. Notice that any two disjoint cycles passing through $u_1$ and $u_2$ give us two disjoint paths between $(s_1,t_1)$ and $(s_2,t_2)$.

\begin{figure}[ht]
	\begin{center}
		\begin{tikzpicture}
		\node (s1) at (0,0) {$s_1$};
		\node (t1) at (0,2) {$t_1$};
		\node (u1) at (-1,1) {$u_1$};
		\node (s2) at (2,0) {$s_2$};
		\node (t2) at (2,2) {$t_2$};
		\node (u2) at (3,1) {$u_2$};
		
		\draw (s1) -- (u1) node [midway, sloped, below] {1};
		\draw (u1) -- (t1) node [midway, sloped, above] {1};
		\draw (s2) -- (u2) node [midway, sloped, below] {1};
		\draw (u2) -- (t2) node [midway, sloped, above] {1};
		\end{tikzpicture}
	\end{center}
	\captionsetup{justification=centering}
	\caption{Converting an instance of $SDP(2)$ to $SDC(2,2)$}
	\label{fig:join}
\end{figure}
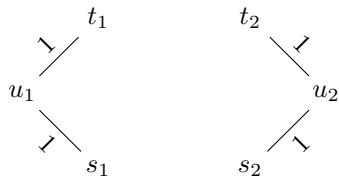

Similarly, connecting the $k$-pairs of vertices via another new vertex and edges of weight $x^0 = 1$, gives us a reduction from $k$-disjoint paths to $k$-disjoint cycles via $k$-vertices. Since this reduction preserves the weight of the path/cycle, it is indeed a reduction from $SDP(k)$ to $SDC(k,k)$.

To apply the techniques of \cite{DBLP:journals/siamcomp/BjorklundH19} to disjoint cycles problem, we instead consider the following variant $SDCE(l,k)$: Given a weighted undirected graph with $k$ marked \textit{edges}, find the minimum of sum of weight of $l$ cycles such that they pass through all the marked edges and are pairwise disjoint.

It can be easily seen that there is a logspace reduction from $SDC(l,k)$ to $SDCE(l,k)$ as follows: Let $(G, \{v_1, v_2, \ldots, v_k\})$ be an instance of $SDC(l,k)$. Assume without loss of generality that the marked vertices form an independent set, or otherwise split the edge into two by introducing a new vertex in the middle. For each $i$, choose a vertex $u_i$, a neighbour of $v_i$, such that $i \neq j \implies u_i \neq u_j$, we solve $(G, \{e_1, e_2, \ldots e_k\})$ an instance of $SDCE(l,k)$ where $e_i = \{u_i, v_i\}$ and output the smallest solution amongst all the instances of $SDCE$ thus created. Since for each $i$, $\deg(v_i) < n$, number of instances of $SDCE$ created are bounded by $O(n^k)$ all of which can be solved in parallel as $k$ is fixed.

\subsection{Pre-processing}

Given a graph $G = (V,E,w)$ and $k$ marked edges $\{e_i = \{s_i,t_i\}\}_{i \leq k}$, assign weight $x^{w(e)}$ to the edge $e$ of $G$ and add self loops (weight $1$) on all vertices except $\{s_i, t_i\}_{i \leq k}$. Observe that all the non-zero terms appearing in the permanent of adjacency matrix correspond to a cycle cover in $G$. To force these $k$-edges in our cycle cover, we direct these edges in a certain way which we shall call as a \textit{pattern}.

Formally, define a pattern $P$ as an ordered pairing of terminals of given edges $\{s_i, t_i \mid 1 \leq i \leq k\}$. Furthermore, we view each undirected edge $\{u,v\}$ in $G$ as two directed edges $(u,v)$ and $(v,u)$ with the same weight. For any pattern $P$, define a pattern graph $G_P$ with the same vertex/edge set as of $G$ but such that if $(u,v) \in P$ then all outgoing edges from $u$, except edge $(u,v)$, are deleted. We denote by $A_P$ the adjacency matrix of $G_P$.

Now we shall show how to solve the $SDCE(1,k)$ problem for any $k \geq 1$. This algorithm essentially follows from the work of \cite{8efe47374b8c4cadb1165092ce46518d}. Next, we also present how to solve the $SDCE(2,k)$ problem for any $k \geq 2$. As far as we know, no algorithm (better than brute force) was known apriori to our work for $k \geq 3$.

\subsection{Shortest Cycle}

Let $\{e_i = \{s_i,t_i\}\}_{i \leq k}$ be given $k$-edges. For each binary sequence $b = (b_1,b_2,\ldots,b_{k-1})$ of length $k-1$, consider the following pattern $P_b$:
\begin{itemize}
	\item $(s_1,t_1) \in P_b$
	\item $\forall 2 \leq i \leq k$, if $b_{i-1}=0$ then $(s_i,t_i) \in P_b$ else $(t_i,s_i) \in P_b$
\end{itemize}

So $\{P_b\}_b$ is the collection of patterns with the orientation of $e_1$ fixed and all possible orientations of the other edges $\{e_i\}_{i\geq 2}$, as dictated by each binary sequence.

\begin{claim}
	Under the assumption that the shortest cycle is unique, the smallest exponent with non-zero coefficient in $f_1(x) \Modulo{2}$ is the weight of unique shortest cycle passing through the given edges, where
	$$f_1(x) = \sum_b \perm(A_{P_b})$$
\end{claim}

\begin{proof}
	Let $C$ be any cycle cover which consists of atleast 2 non-trivial cycles. Consider the cycle in $C$ which doesn't contain edge $e_1$ - there are two ways of orienting this cycle, namely clockwise and counter-clockwise. So this cycle cover contributes to $f_1(x)$ for atleast two such $b$-sequences and so it vanishes modulo $f_1 \Modulo{2}$.
	
	Thus the only terms that survive in $f_1 \Modulo{2}$ are the cycle covers which consist of one cycle passing through all the given edges and self-loops on the remaining vertices, and furthermore number of cycles of this weight must be odd.
	
	Since the shortest weight cycle was unique by our assumption, we get the desired result.
\end{proof}

To drop the assumption that a unique minimum weight solution exists, we instead assign modified weights $2nmw(e) + w'(e)$ where $n = |V(G)|$, $m = |E(G)|$, $w(e)$ is the given weight of edge $e$ and $w'(e) \in \{0,2, \ldots, 2m-1\}$ is choosen independently and uniformly at random for each edge $e$. Then isolation lemma \cite{mvv} tells us that, with probability $\frac12$, the minimum weight cycle is unique. Hence if the term $x^j$ survives as the smallest exponent with non-zero coefficient term in $f_1(x) \Modulo{2}$ then we get the weight of shortest cycle as $\floor{j/2nm}$.

\subsection{Shortest 2-Disjoint Cycles}

We shall first prove the following stronger result:

\begin{theorem}
	\label{th:separatingedges}
	Given a set of $k$-edges $\{e_i\}_{i \leq k}$, we can find weight of the shortest 2-disjoint cycles passing through these edges such that $e_1$ and $e_2$ appear in different cycles in $\parity \L / \poly$ (and $\RNC$)
\end{theorem}

Let $\{P_b\}_b$ be the patterns as defined above. Furthermore, for each binary sequence $c=(c_1,c_2,\ldots,c_{k-2})$ of length $k-2$, define pattern $Q_c$ as
\begin{itemize}
	\item $(s_1,s_2) \in Q_c$
	\item $(t_1,t_2) \in Q_c$
	\item $\forall 3 \leq i \leq k$, if $c_{i-2}=0$ then $(s_i,t_i) \in Q_c$ else $(t_i,s_i) \in Q_c$
\end{itemize}

With a combination of these patterns, we can get our desired cycle covers. We claim that the non-zero terms appearing in $$f_2(x) = \sum_b \perm(A_{P_b}) - \sum_c \perm(A_{Q_c})$$
correspond to cycle covers in $G_{P_b}$ such that edges $e_1$ and $e_2$ appear in different cycles.

To prove our claim, we need to argue that the cycle covers of $G_{P_b}$ in which $e_1$ and $e_2$ appear in the same cycle are exactly the cycle covers of $G_{Q_c}$ Let

$$\mathcal{C}_Q = \bigsqcup_c \text{ cycle covers of }G_{Q_c}$$
$$\mathcal{C}_P = \bigsqcup_b \text{ cycle covers of }G_{P_b} \text{ such that edges $e_1$ and $e_2$ appear in the same cycle}$$

where each cycle cover is counted with repetitions in $\mathcal{C}_P$ and $\mathcal{C}_Q$.

\begin{claim}
	\label{lemma:cpcqbijection}
	There is a one-one correspondence between $\mathcal{C}_P$ and $\mathcal{C}_Q$.
\end{claim}

\begin{proof}
	We define the mapping $\varphi: \mathcal{C}_P \to \mathcal{C}_Q$ as follows. Given a cycle cover in $\mathcal{C}_P$, remove the edges $e_1$ and $e_2$ and add edges $(s_1, s_2)$ and $(t_1, t_2)$, refer to Figure \ref{fig:bijection} below.
	
	To show that this a well-defined mapping and indeed a bijection, we partition $\mathcal{C}_P$ into type 1 and type 2 cycle covers, depending upon the orientation of the edge $e_2$. Consider the cycle in which $e_1$ and $e_2$ appear together. Then if the edge $e_2$ is oriented as $(s_2,t_2)$ then we call it type 1 cycle cover otherwise we call it a type 2 cycle cover. 
	
	Similarly, we partition $\mathcal{C}_Q$ into type 1 and type 2 cycle covers. If the edges $\{s_1,s_2\}$ and $\{t_1,t_2\}$ appear in the same cycle then we call it a type 1 cycle cover otherwise we call it a type 2 cycle cover.
	
	Fix a type 1 cycle cover of $\mathcal{C}_P$. Then it contains a cycle of the form $$(s_1 \xrightarrow{e_1} t_1 \rightarrow P_1 \rightarrow s_2 \xrightarrow{e_2} t_2 \rightarrow P_2 \rightarrow s_1)$$ 
	Applying $\varphi$ to this cycle cover we get the cycle $$(s_1 \xrightarrow{e_1} s_2 \rightarrow P_1^{\text{reverse}} \rightarrow t_1 \xrightarrow{e_2} t_2 \rightarrow P_2^{\text{reverse}} \rightarrow s_1)$$ and the other cycles remain intact. This constitutes a type 1 cycle cover in $\mathcal{C}_Q$
	
	Similarly, consider a type 2 cycle cover of $\mathcal{C}_P$ with the cycle $$(s_1 \xrightarrow{e_1} t_1 \rightarrow P_1 \rightarrow t_2 \xrightarrow{e_2} s_2 \rightarrow P_2 \rightarrow s_1)$$
	Applying $\varphi$ to this cycle cover we get two cycles $$(s_1 \xrightarrow{e_1} s_2 \rightarrow P_2^{\text{reverse}} \rightarrow  s_1)$$ $$(t_1 \xrightarrow{e_2} t_2 \rightarrow P_1^{\text{reverse}} \rightarrow t_1)$$ and the other cycles remain intact. This constitutes a type 2 cycle cover in $\mathcal{C}_Q$.
	
	Therefore, $\varphi$ is a well-defined mapping and furthermore type $i$ cycle covers of $\mathcal{C}_P$ are mapped to type $i$ cycle covers of $\mathcal{C}_Q$, $i \in \{1,2\}$. Now consider $\psi: \mathcal{C}_Q \to \mathcal{C}_P$ defined as follows. Given a cycle cover in $\mathcal{C}_Q$, remove the edges $(s_1, s_2)$ and $(t_1, t_2)$ in the cycle and insert edges $e_1$ and $e_2$ with the orientation decided by the type. By an similar argument as above, we get that $\psi$ is well-defined and clearly $\psi$ is inverse of $\varphi$.
	
	\begin{figure}[htb]
		\begin{center}
			\begin{tikzpicture}
			\node (psi) at (0,0) {$s_1$};
			\node (pti) at (0,1) {$t_1$};
			\node (psj) at (2,0) {$s_2$};
			\node (ptj) at (2,1) {$t_2$};
			\node (qsi) at (5,0) {$s_1$};
			\node (qti) at (5,1) {$t_1$};
			\node (qsj) at (7,0) {$s_2$};
			\node (qtj) at (7,1) {$t_2$};
			
			\draw[->] (psi) -- (pti);
			\draw[->] (psj) -- (ptj);
			\draw[dashed] (pti) -- (psj);
			\draw[dashed] (psi) -- (ptj);
			\draw[<->, thick] (3,0.5) -- (4,0.5);
			\draw[->] (qsi) -- (qsj);
			\draw[->] (qti) -- (qtj);
			\draw[dashed] (qti) -- (qsj);
			\draw[dashed] (qsi) -- (qtj);
			
			\node at (1,-1) {Type 1 cycle cover in $\mathcal{C}_P$};
			\node at (6,-1) {Type 1 cycle cover in $\mathcal{C}_Q$};
			\end{tikzpicture}
		\end{center}
		\quad\quad
		\begin{center}
			\begin{tikzpicture}
			\node (psi) at (0,0) {$s_1$};
			\node (pti) at (0,1) {$t_1$};
			\node (psj) at (2,0) {$s_2$};
			\node (ptj) at (2,1) {$t_2$};
			\node (qsi) at (5,0) {$s_1$};
			\node (qti) at (5,1) {$t_1$};
			\node (qsj) at (7,0) {$s_2$};
			\node (qtj) at (7,1) {$t_2$};
			
			\draw[->] (psi) -- (pti);
			\draw[->] (ptj) -- (psj);
			\draw[dashed] (psi) -- (psj);
			\draw[dashed] (pti) -- (ptj);
			\draw[<->, thick] (3,0.5) -- (4,0.5);
			\draw[->] (qsi) -- (qsj);
			\draw[->] (qti) -- (qtj);
			\draw[dashed] (qsi) .. controls (6,0.5) .. (qsj);
			\draw[dashed] (qti) .. controls (6,0.5) .. (qtj);
			
			\node at (1,-1) {Type 2 cycle cover in $\mathcal{C}_P$};
			\node at (6,-1) {Type 2 cycle cover in $\mathcal{C}_Q$};
			\end{tikzpicture}
		\end{center}
		\captionsetup{justification=centering}
		\caption{Bijection between $\mathcal{C}_P$ and $\mathcal{C}_Q$}
		\label{fig:bijection}
	\end{figure}
\end{proof}

Now suppose $C$ is a cycle cover of $G$ such that edges $e_1$ and $e_2$ appear in different cycles. We have two cases:

\textbf{Case 1}: number of non-trivial cycles in $C$ is more than $2$. Consider any two cycles in $C$ such that $e_1$ is not incident on them. We can orient these two cycles in both clockwise and anti-clockwise direction and so we get that $C$ is a cycle cover in $G_{P_b}$ for atleast $4$ $b$-sequences. Hence, the term corresponding to $C$ cancels out in $f_2 \Modulo{4}$

\textbf{Case 2}: number of non-trivial cycles is exactly two. Then $C$ is a cycle cover in $G_{P_b}$ for exactly two $b$-sequences , that is the the cycle passing through $e_2$ has two possible orientations whereas cycle passing through $e_1$ has a fixed orientation (as orientation of $e_1$ remains fixed in all $G_{P_b}$). Hence, the term corresponding to $C$ appears with a coefficient of two in $f_2 \Modulo{4}$. Therefore, the non-zero terms in $f_2 \Modulo{4}$ correspond to only the cycle covers in which edges $e_1$ and $e_2$ appear in different cycles and number of non-trivial cycles is exactly $2$. Assuming a unique shortest 2-disjoint cycle exists, it's weight can be obtained from the smallest exponent with a non-zero coefficient in $f_2 \Modulo{4}$. Finally, to drop this assumption, we again assign random weights as done previously to ensure that the minimum weight solution is unique, with high probability. In the next section we discuss a common weighting scheme to obtain a $\parity \L / \poly$ algorithm. This shall complete the proof of Theorem \ref{th:separatingedges}.

\begin{corollary} (Theorem \ref{thm:main2} restated)
	Given a set of $k$-edges $\{e_i\}_{i \leq k}$, we can find weight of the shortest 2 cycles passing through these edges in $\parity \L / \poly$ (and $\RNC$)
\end{corollary}

\begin{proof}
	For each pair of edges $e_i$ and $e_j$, we can find weight of the shortest cycles separating them using the above algorithm. Hence taking the minimum over all pairs, we get our desired result.
\end{proof}

\subsection{Common Weighting Scheme}

We have only exhibited a randomized $\parity \L$ algorithm (that is $\RNC$ algorithm). To further show that a common poly weight scheme exists for all graphs of size $n$, we use the well-known result of \cite{DBLP:journals/siamcomp/ReinhardtA00} as follows: Call a weighted undirected graph $(G,w)$ ($w$ is the given weight function on edges) \textit{min-$k$-unique}, if for any $k$ marked edges on $G$, there exists unique shortest $l$ disjoint cycles passing through these $k$ edges. Our goal is to show for each $n > 0$ there exists a set of $n^2$ weight functions $w_1,\ldots,w_{n^2}$ such that given a graph $G$ on $n$ vertices, $(G,w_i)$ is \textit{min-$k$-unique} for some $i \in [1,n^2]$.

Given a graph $G$ on $n$ vertices and $k$ marked edges $e_1, \ldots, e_k$, let $\mathcal{F}(e_1,\ldots,e_k)$ be the family of all $l$ disjoint cycles passing through $e_1,\ldots,e_k$. Using isolation lemma \cite{mvv}, if $w$ is a random weight function, that is each edge is assigned a weight from $[1,4n^{2k+2}]$ independently and uniformly at random, then probability that $\mathcal{F}(e_1,\ldots,e_k)$ has a unique minimum weight element is atleast $1 - 1/4n^{2k}$. Therefore, probability that $(G,w)$ is not \textit{min-$k$-unique} for a random weight function $w$ is atmost 

\begin{align*}
\Pr[\exists e_1,\ldots,e_k: \mathcal{F}(e_1,\ldots,e_k) \text{ doesn't have a minimum weight element}] \\ 
\leq \sum_{e_1,\ldots,e_k} \frac{1}{4n^{2k}} \leq 1/4
\end{align*}

Now we claim that there exists a set of $n^2$ weight functions $W = (w_1,\ldots,w_{n^2})$ such that for any given graph $G$ on $n$ vertices, $(G,w_i)$ is \textit{min-$k$-unique} for some $1 \leq i \leq n^2$. We say \textit{$W$ is bad} if it doesn't meet this criteria and in particular \textit{$W$ is bad for $G$}, if none of $(G, w_i)$ is \textit{min-$k$-unique}. For a randomly choosen $W$, that is each $w_i$ is chosen independently and uniformly at random, then

$$\Pr[W \text{ is bad for } G] \leq \Pr[\forall i: (G,w_i) \textit{ is not min-$k$-unique}] \leq \left(\frac14\right)^{n^2}$$
$$\implies \Pr[W \text{ is bad}] \leq \Pr[\exists G: W \text{ is bad for } G] \leq 2^{n^2}\left(\frac14\right)^{n^2} < 1$$

Hence there exists some $W = (w_1,\ldots,w_{n^2})$ which satisfies the above property and so $W$ is the required poly advice. 

To complete the argument for $SDCE(1,k), SDCE(2,k) \in \parity \L / \poly$, we obtain the weight of shortest cycle(s), by replacing replace $w'$, the random weight function, with each of the weight functions $w_i$ and output the minimum amongst them.

\subsection{Constructing Cycles}

We remark that under the assumption that the shortest cycle(s) are unique, we can recover these cycles $C$ just from the knowledge of their weight $w(C)$. This follows the standard strategy of solving search via isolation as in \cite{mvv}. For each edge $e \not \in \{e_1, \ldots, e_k\}$, delete the edge $e$ and call the resulting graph $G_e$. Running our algorithm on $(G_e, \{e_1, \ldots, e_k\})$, if the shortest cycle(s) weight is more than $w(C)$, then discard $e$ otherwise add $e$ to the set $C$, which gives us the required cycle(s).

\section{Hafnians and counting perfect matchings modulo $2^k$}
\label{sec:hafnian}

Similar to permanent and determinant, another pair of well-studied algebraic analogous functions on a matrix are hafnians and pfaffians. Let $A = (a_{ij})$ be a symmetric $2n \times 2n$ matrix over integers, hafnian is defined as
\begin{align}
\hf(A) &= \frac{1}{2^nn!} \sum_{\sigma \in S_{2n}} \prod_{j=1}^n a_{\sigma(2j-1),\sigma(2j)}
\end{align}

Note that the diagonal entries of $A$ don't contribute in the calculation of hafnians and hence we can assume them to be $0$. Let $B = (b_{ij})$ be a skew-symmetric $2n \times 2n$ matrix, pfaffian is defined as
\begin{align}
\pf(B) &= \frac{1}{2^nn!} \sum_{\sigma \in S_{2n}} \sgn(\sigma) \prod_{j=1}^n b_{\sigma(2j-1),\sigma(2j)}
\end{align}

But notice that $\hf(A) \equiv \pf(A) \Modulo{2}$. \cite{DBLP:journals/dam/MahajanSV04} have shown that $\pf(A)$ can be computed in $\NC$ and hence as an immediate consequence we get that $\hf(A) \pmod{2}$ can be computed in $\NC$. We can reduce the computation of hafnian to several hafnians of smaller submatrices using the following lemma. Denote by $A[i,j]$ the matrix obtained from $A$ by deleting rows $i$ and $j$, columns $i$ and $j$.

\begin{lemma} (\cite{DBLP:journals/algorithmica/HiraiN18} Lemma 2.2)
	$$\hf(A) = \sum_{j: j\neq i} a_{ij} \hf(A[i,j])$$
	$$\hf(A) = a_{ij} \hf(A[i,j]) + \sum_{pq: p,q \not \in \{i,j\}, p \neq q} (a_{ip}a_{jq}+a_{iq}a_{jp}) \hf(A[i,j,p,q])$$
\end{lemma}

Assume $\pf(A) \equiv 0 \pmod{2}$, then $\det(A) \equiv 0 \pmod{2}$ and we can find a vector $v \in \mathbb{Z}_2^{2n}$ such that $Av = A^Tv = 0 \pmod{2}$. Assume without loss of generality $v_1 = 1$.

Let $r_i, c_i$ denote the $i^{th}$ row and $i^{th}$ column of $A$ respectively. 

\begin{itemize}
	\item Construct $A'$  by replacing first row with $\sum v_i r_i$ and then replacing first column with $\sum v_i c_i$ 
	\item Construct $A_i$ by replacing first row with $r_i$ and first column with $c_i$.
\end{itemize}

Then we check that
\begin{align*}
\hf(A') &= \sum_{j > 1} \left(\sum_{i \geq 1} v_i a_{ij} \right) \hf(A[1,j]) \\
&= \sum_{i \geq 1} v_i \left( \sum_{j > 1} a_{ij} \hf(A[1,j]) \right) \\
&= \sum_{j > 1} a_{1j} \hf(A[1,j]) + \sum_{i > 1} v_i \left( \sum_{j > 1} a_{ij} \hf(A[1,j]) \right) \\
&= \hf(A) + \sum_{i > 1} v_i \hf(A_i)
\end{align*}

\textbf{Computing $\hf(A')$}: since $A^Tv = 0 \Modulo{2} \implies \sum_{i \geq 1} v_i a_{ij} = 2b_j \Modulo{2^k}$ for some $c_j \in \mathbb{Z}$ and hence
\begin{align*}
\hf(A') &= \sum_{j > 1} \left(\sum_{i \geq 1} v_i a_{ij} \right) \hf(A[1,j]) \\
&= \sum_{j > 1} 2b_j \hf(A[1,j]) \\
\implies \hf(A') \Modulo{2^k} &=  2\left(\sum_{j > 1} b_j \hf(A[1,j]) \Modulo{2^{k-1}}\right)
\end{align*}

\textbf{Computing $\hf(A_i)$}:
\begin{align*}
\hf(A_i) &= a_{ii} \hf(A[1,i]) + \sum_{pq: p,q \not \in \{1,i\}, p \neq q} 2a_{ip}a_{iq} \hf(A[1,i,p,q]) \\
\implies \hf(A_i) \Modulo{2^k} &= 2\left(\sum_{pq: p,q \not \in \{1,i\}, p \neq q} a_{ip}a_{iq} \hf(A[1,i,p,q]) \Modulo{2^{k-1}} \right)
\end{align*}

Thus, we can compute $\hf(A) \Modulo{2^k}$ provided $\pf(A) \equiv 0 \Modulo{2}$. 

Now if $\pf(A) \not \equiv 0 \Modulo{2}$, then we can find $(i,j), i \neq j$ such that $\hf(A[i,j]) \not \equiv 0 \Modulo{2}$. Consider the matrix $C$ where all entries are same as in $A$ except $a_{ij}$ is replaced with $a_{ij}+1$, then we get $\hf(C) = \hf(A) + \hf(A[i,j])$. Since $\hf(C) \equiv 0 \Modulo{2}$, we can compute $\hf(C) \Modulo{2^k}$ as described above and since $\hf(A[i,j])$ is a $(n-2) \times (n-2)$ matrix, we compute it's hafnian recursively modulo $2^k$. Therefore, we can compute $\hf(A) = \hf(C) - \hf(A[i,j]) \Modulo{2^k}$.

This gives us a $\P$ algorithm for computing hafnians modulo $2^k$.

\paragraph*{Counting perfect matchings modulo $2^k$}

Let $G$ be an undirected graph and $A_G$ denote the adjacency matrix of the graph $G$. If $G$ has odd number of vertices, then clearly there aren't any perfect matchings. Otherwise it is straight-forward to see that number of perfect matchings in $G$ is same as the value $\hf(A_G)$. Hence the result follows.

	\section{Conclusion}
	\label{sec:conclusion}
	
	We started by recognizing the appropriate algebraic structure $\mathfrak{R}$ over which we can present a parallel algorithm to compute permanent modulo $2^k$. Then we saw two techniques to get permanent over $\mathbb{Z}[x] \Modulo{2^k}$ from $\mathfrak{R} \Modulo{2^k}$. First method was to choose a large enough irreducible polynomial for our ring $\mathfrak{R}$. Another method was to interpolate over the ring $\mathfrak{R}$, which was an extension of the commonly known interpolation over finite fields.
	
	Then we considered applications for parallel polynomial permanent. This includes a direct parallelization of the shortest 2-disjoint paths problem as given by \cite{DBLP:journals/siamcomp/BjorklundH19}. Another direct application, although which required some modification, was finding shortest cycle passing through given vertices \cite{8efe47374b8c4cadb1165092ce46518d}. We further presented a common framework to view the above mentioned problems. This also aided us in further generalizing and obtaining a parallel algorithm to find shortest 2-disjoint cycles in any weighted undirected graph.
	
	The more general question of computing permanent over arbitrary commutative rings of characteristic $2^k$ for $k \geq 2$ still remains open. On the other hand, using the framework we presented, we ask if it is possible to obtain shortest $k$ disjoint cycles for $k \geq 3$?
	
	\bibliography{permanent}
\end{document}